\documentclass{article}
\usepackage{amsmath,amsthm,amssymb}
\usepackage{array, longtable}
\newtheorem{Theorem}{Theorem}[section]
\newtheorem{lem}[Theorem]{Lemma}
\newtheorem{Remark}[Theorem]{Remark}
\newtheorem{Definition}[Theorem]{Definition}
\newtheorem{Corollary}[Theorem]{Corollary}
\newtheorem{Proposition}[Theorem]{Proposition}
\newtheorem{Example}[Theorem]{Example}

\numberwithin{equation}{section}
\numberwithin{table}{section}

\begin{document}

\title{Polyadic Constacyclic Codes}

\insert\footins{\footnotesize {\it Email addresses}:
{bocong\_chen@yahoo.com (B. Chen)},
 hdinh@kent.edu (H. Q. Dinh),
  yfan@mail.ccnu.edu.cn (Y. Fan),
  lingsan@ntu.edu.sg (S. Ling). }

\author{Bocong Chen$^{\bf 1,2}$,~ Hai Q. Dinh$^{\bf 3}$,~
        Yun Fan$^{\bf 1}$,~ San Ling$^{\bf 2}$}

\date{\small $^{\bf 1}$ School of Mathematics and Statistics,
 Central China Normal University, Wuhan 430079, China\\
 $^{\bf 2}$ School of Physical \& Mathematical Sciences,
 Nanyang Technological University, Singapore 637616, Singapore\\
 $^{\bf 3}$ Department of Mathematical Sciences, Kent State University,
 4314 Mahoning Avenue, Warren, OH 44483, USA}

\maketitle

\begin{abstract} For any given positive integer $m$,
a necessary and sufficient condition for the existence
of Type I {$m$-adic} constacyclic codes is given. Further,
for any given integer $s$,
a necessary and sufficient condition for $s$ to be a multiplier
of a Type {\rm I} polyadic constacyclic code is given.
As an application,
some optimal codes from Type I polyadic constacyclic codes,
including generalized Reed-Solomon codes and alternant MDS codes,
are constructed.


\end{abstract}

\section{Introduction}
The class of duadic cyclic codes over finite fields, which includes
the important family of quadratic residue codes, was introduced by
Leon {\it et al.}  \cite{Leon}, and then studied  by several authors
such as {in \cite{Pless}, \cite{Smid}, \cite{Ding}, \cite{Ding2} and
\cite{Han}.

Motivated by the good properties of duadic cyclic codes, Pless and
Rushanan \cite{Pless2} moved on to study  triadic cyclic codes. The
class of polyadic cyclic codes (or $m$-adic cyclic codes) was later
introduced by Brualdi and Pless \cite{Brualdi}. Subsequently,
Rushanan {\it et al.} generalized duadic cyclic codes to duadic
abelian codes (\cite{Rushanan}, \cite{Ward}). Zhu {\it et al.}
further studied duadic group algebra codes (\cite{Zhu},
\cite{Zhang}, \cite{Aly}). Results on the existence conditions of
these codes have been obtained by these authors. Ling and Xing
\cite{Ling04} extended the definition of polyadic cyclic codes to
include noncyclic abelian codes, and obtained necessary and
sufficient conditions for the existence of nondegenerate polyadic
codes; some interesting examples arising from this family of codes
were also given. Sharma {\it et al.} \cite{Sharma} removed the
``nondegenerate" condition considered by Ling and Xing in
\cite{Ling04}, and  determined necessary and sufficient conditions
for the existence of polyadic cyclic codes of prime power length.

Another direction of generalization for the notion of duadic cyclic
codes is the study of polyadic constacyclic codes over finite
fields.  In \cite{Lim}, polyadic cyclic codes were generalized to
polyadic consta-abelian codes, and some sufficient conditions for
the existence of this class of codes were established. Duadic
negacyclic codes, which is a special class of polyadic constacyclic
codes, were considered by Blackford \cite{Blackford08}. Recently,
Blackford \cite{Blackford13} continued to study Type {\rm I} duadic
constacyclic codes (see Definition~\ref{def-adic} for
detail). Necessary and sufficient conditions for the existence of
Type {\rm I} duadic constacyclic codes were given and, for a given
integer $s$, equivalent conditions were also obtained to determine
whether or not $s$ can be a multiplier for a Type {\rm I} duadic
code. However, to the best of our knowledge, there are no known
solutions to the following general questions: for $n$ a positive
integer and $\lambda$ a nonzero element of the underlying field,
\begin{enumerate}

\item For any given positive integer $m$,
do Type~I $m$-adic $\lambda$-constacyclic codes of length $n$ exist?

\item For any given integer $s$, can $s$ be a multiplier of a
Type I polyadic $\lambda$-constacyclic code of length $n$?
\end{enumerate}

In this paper, a necessary and sufficient condition for the
existence of Type {\rm I} $m$-adic $\lambda$-constacyclic codes is
given. Further, a necessary and sufficient condition for $s$ to be a
multiplier of a Type {\rm I} polyadic $\lambda$-constacyclic code is
obtained. We also exhibit some optimal polyadic constacyclic codes,
including generalized Reed-Solomon codes and alternant MDS codes.

This paper is organized as follows. In Section 2,
basic notations and the main results of this paper are presented.
In Section 3, we prove some lemmas which play important roles
in the proofs of the main results.
In Section 4, the proofs of the main results are given.
In Section 5, several corollaries are derived from the main results,
and some optimal codes from Type I polyadic constacyclic codes
are constructed, including generalized Reed-Solomon codes and
alternant MDS codes.

\section{Notations and main results}
We denote by $\mathbb{F}_q$ the finite field with cardinality $|\mathbb{F}_q|=q$.
Let $\lambda \in \mathbb{F}_q^*$,  where $\mathbb{F}_q^*$ denotes
the multiplicative group of units of ${\Bbb F}_q$,
and let $n$ be a positive integer coprime to $q$.
Any ideal $C$ of the quotient ring
$\mathbb{F}_q[X]/\langle X^n-\lambda\rangle$ is said to be
a {\em $\lambda$-constacyclic code} over ${\Bbb F}_q$ of length $n$.
Let ${\rm ord}_{\mathbb{F}_q^*}(\lambda)=r$,
where ${\rm ord}_{\mathbb{F}_q^*}(\lambda)$ denotes the order of
$\lambda$ in the multiplicative group $\mathbb{F}_q^*$;
then $r\mid (q-1)$ since $\mathbb{F}_q^*$ is a cyclic group
of order $q-1$.
Thus a constacyclic code $C$ has three parameters $q,n,r$;
in this case we say that $C$ is a {\em $(q,n,r)$-constacyclic code}.
In this paper we always adopt the following notations:
\begin{itemize}
\item
$q$, $n$, $r$ with $\gcd(q,n)=1$ and $r|(q-1)$
are the parameters of the constacyclic code $C$,
where $\gcd(q,n)$ denotes the greatest common divisor;
\item
${\Bbb Z}_{rn}$ denotes the residue ring
of the integer ring ${\Bbb Z}$ modulo $rn$;
\item
${\Bbb Z}_{rn}^*$ denotes the multiplicative group
consisting of units of ${\Bbb Z}_{rn}$;
\item
$1+r{\Bbb Z}_{rn}
=\{1+rk\mid k=0,1,\cdots,n-1\} \subseteq {\Bbb Z}_{rn}$;
\item
$\mu_h$, where $\gcd(h,rn)=1$,
denotes the permutation of the set ${\Bbb Z}_{rn}$ given by
$\mu_h(x)=hx$ for $x\in{\Bbb Z}_{rn}$;
\item
$s$ is an integer such that $s\in{\Bbb Z}_{rn}^*\cap(1+r{\Bbb Z}_{rn})$,
and $m$ is a positive integer.
\end{itemize}

Let $e$ be the multiplicative order of  $q$ modulo $rn$, i.e.,
$rn\mid(q^e-1)$ but $rn\nmid(q^{e-1}-1)$.
Then, in the finite field $\mathbb{F}_{q^e}$, there is a primitive
$rn$th root $\omega$ of unity such that $\omega^n=\lambda$.
It is easy to check the following facts:
\begin{itemize}
\item[$\blacktriangleright$]
{\it $\omega^i$, $i\in(1+r{\Bbb Z}_{rn})$, are just all the roots of $X^n-\lambda$.}
\item[$\blacktriangleright$]
{\it
For an integer $h$ coprime to $rn$, the set $1+r{\Bbb Z}_{rn}$ is $\mu_h$-invariant
if and only if $h\in{\Bbb Z}_{rn}^*\cap(1+r{\Bbb Z}_{rn})$.}
\end{itemize}
Since $\gcd(q,n)=1$ and $r|(q-1)$,
it follows that $q\in {\Bbb Z}_{rn}^*\cap(1+r{\Bbb Z}_{rn})$ and
$1+r{\Bbb Z}_{rn}$ is $\mu_q$-invariant.
Let $(1+r{\Bbb Z}_{rn})/\mu_q$ denote the set of $\mu_q$-orbits
within $1+r{\Bbb Z}_{rn}$,
i.e., the set of $q$-cyclotomic cosets within~$1+r{\Bbb Z}_{rn}$.
For any $q$-cyclotomic coset $Q$ in ${\Bbb Z}_{rn}$,
the polynomial $M_Q(X)=\prod_{i\in Q}(X-\omega^i)$ is irreducible
in $\mathbb{F}_q[X]$. Thus
$$
X^n-\lambda=\prod_{Q\in(1+r{\Bbb Z}_{rn})/\mu_q}M_{Q}(X)
$$
is the monic irreducible decomposition of $X^n - \lambda$  in $\mathbb{F}_q[X]$.

\begin{Definition}\label{def-adic}\rm
If $1+r{\Bbb Z}_{rn}$ has a partition $1+r{\Bbb Z}_{rn}={\cal
X}_0\cup\cdots\cup{\cal X}_{m-1}$ such that, for some integer $s$,
every ${\cal X}_j$ is $\mu_q$-invariant and $\mu_s({\cal X}_j)={\cal
X}_{j+1}$ for $j=0,1,\cdots,m-1$ (the subscripts are taken modulo
$m$), then

\begin{itemize}

\item[(i)] the partition
$\big\{{\cal X}_0,{\cal X}_1,\cdots,{\cal X}_{m-1}\big\}$
is called a {\em Type I $m$-adic splitting} of $1+r{\Bbb Z}_{rn}$, and
$\mu_s$ is said to be a {\em multiplier} of the
Type I $m$-adic splitting;

\item[(ii)] the constacyclic codes $C_{{\cal X}_j}$, with
check polynomial $\prod_{Q\in{\cal X}_j/\mu_q}M_{Q}(X)$
for $j=0,1,\cdots,m-1$, are called
{\em Type I $m$-adic constacyclic codes} given by the multiplier $\mu_s$.

\end{itemize}
\end{Definition}

{\bf Remark.}~
In the case of Definition~\ref{def-adic}, the following map
(denoted by $\hat\mu_s$):
$$\hat\mu_s:\quad {\Bbb F}_q[X]/\langle X^n-\lambda\rangle \longrightarrow
{\Bbb F}_q[X]/\langle X^n-\lambda\rangle,\quad
\sum_{i=0}^{n-1}a_iX^i\longmapsto \sum_{i=0}^{n-1}a_iX^{is},
$$
is an isometry (i.e., $\hat\mu_s$ keeps both the algebraic structure
and the weight structure, see \cite[\S3]{CFLL})
of the quotient algebra ${\Bbb F}_q[X]/\langle X^n-\lambda\rangle$
such that
$\hat\mu_s(C_{{\cal X}_j})=C_{{\cal X}_{j+1}}$
for $j=0,1,\cdots,m-1$ and
${\Bbb F}_q[X]/\langle X^n-\lambda\rangle=\bigoplus_{j=0}^{m-1}C_{{\cal X}_j}$.
The map $\hat\mu_s$ is also called a {\em multiplier} of the quotient algebra,
e.g., see \cite[Theorem 4.3.12]{Huffman}.

\medskip
Of course, the notion of an $m$-adic splitting
  makes sense in practice only for $m>1$.
We allow that $m=1$ since it is convenient for the
statements of our results.
As mentioned in Section $1$, there are two fundamental
questions concerning Type~{\rm I} $m$-adic constacyclic codes:
\begin{itemize}
\item
Under what conditions do Type I $m$-adic $(q,n,r)$-constacyclic codes exist?
\item
For a given integer $s$, is $\mu_s$ a multiplier of a Type I $m$-adic splitting for $1+r{\Bbb Z}_{rn}$?
\end{itemize}
We address these two questions in this paper.

\medskip
Let $t$ be a non-zero integer. For any prime $p$, there is a unique
non-negative integer $\nu_p(t)$ such that $p^{\nu_p(t)}\Vert t$,
 i.e., $p^{\nu_p(t)}$ is the largest power of $p$
dividing $t$. The function $\nu_p(t)$ is well known as the {\em
$p$-adic valuation of $t$.} Of course,
$t=\pm\prod_{p}p^{\nu_p(t)}$, where $p$ runs over all primes, but
$\nu_p(t)=0$ for all except finitely many primes~$p$. We adopt the
convention that $\nu_p(0)=-\infty$ and $|\nu_p(0)|=\infty$.

The following two theorems are the main results of this paper.

\begin{Theorem}\label{theorem1}
There is a unique integer $M=\prod_{p}p^{\nu_p(M)}$ such that
 Type I $m$-adic $(q,n,r)$-constacyclic codes  exist
if and only if $m$ is a divisor of $M$,
where $\nu_p(M)$ is determined as follows:
if $p\nmid r$ or $p\nmid n$, then $\nu_p(M)=0$; otherwise:
\begin{itemize}
\item[\bf(i)] if $p$ is odd or $\nu_p(r)\ge 2$, then
      $\nu_p(M)=\min\{\nu_p(q-1)-\nu_p(r),\,\nu_p(n)\};$
\item[\bf(ii)] if $p=2$ and $\nu_2(r)=1$, there are two subcases:
 \begin{itemize}
  \item[\bf(ii.1)]
    if $\nu_2(q-1)\ge 2$, then
    $\nu_2(M)=\max\{\min\{\nu_2(q-1)-2,\,\nu_{2}(n)-1\},\,1\};$
  \item[\bf(ii.2)]
    if $\nu_2(q-1)=1$, then $\nu_2(M)=\min\{\nu_2(q+1)-1,\,\nu_{2}(n)-1\}$.
 \end{itemize}
\end{itemize}
\end{Theorem}

\medskip
Note that
$\nu_2(q-1)=1$ if and only if $q\equiv -1\pmod 4$,  which is equivalent to
$\nu_2(q+1)\ge 2$.

\begin{Theorem}\label{theorem2} There is a unique integer
$M_s=\prod_{p}p^{\nu_p(M_s)}$ such that $\mu_s$ is a Type I $m$-adic
splitting  for $1+r{\Bbb Z}_{rn}$ if and only if $m$ is a divisor of
$M_s$, where $\nu_p(M_s)$ is determined as follows: if $p\nmid r$ or
$p\nmid n$, then $\nu_p(M_s)=0$; otherwise:
\begin{itemize}
\item[\bf(i)]
if $p$ is odd or
$p=2$ and both $\nu_p(q-1)\ge 2$ and $\nu_p(s-1)\ge 2$ hold, then
 $$\nu_p(M_s)=\max\big\{\min\{\nu_p(q-1),\,\nu_p(rn)\}-|\nu_p(s-1)|,\,0\big\};$$
\item[\bf(ii)]
if $p=2$, $\nu_2(q-1)=1$ and  $\nu_2(s-1)\ge 2$, then
 $$\nu_2(M_s)=\max\big\{\min\{\nu_2(q+1)+1,\,\nu_2(rn)\}-|\nu_2(s-1)|,\,0\big\};$$
\item[\bf(iii)]
if $p=2$, $\nu_2(q-1)\ge 2$ and $\nu_2(s-1)=1$, then
 $$\nu_2(M_s)=\max\big\{\min\{\nu_2(q-1),\,\nu_2(rn)\}-|\nu_2(s+1)|,\,1\big\};$$
\item[\bf(iv)]
if $p=2$, $\nu_2(q-1)=1$ and $\nu_2(s-1)=1$, then
$$\hskip-9pt\nu_2(M_s)=\left\{\hskip-3pt\begin{array}{r}\max\big\{\hskip-2pt
 \min\{\nu_2(q+1)+1,\nu_2(rn)\}
    -\min\{|\nu_2(s+1)|,\nu_2(q+1)\},0\big\}, \\
  {\rm if}~~ \nu_2(s+1)\ne\nu_2(q+1);\\[2mm]
  0, \hfill{\rm if}~~ \nu_2(s+1)=\nu_2(q+1).
\end{array}\right.$$
\end{itemize}
\end{Theorem}

\section{Preparations}
Let ${\cal X}$ be a finite set and let ${\rm Sym}({\cal X})$ be the symmetric group of ${\cal X}$  consisting of all permutations of ${\cal X}$.
If  $\mu\in {\rm Sym}({\cal X})$, i.e., $\mu$ is a permutation of  ${\cal X}$, then
$\langle\mu\rangle=\{\mu^j\,|\,j\in \mathbb{Z}\}$  acts on ${\cal X}$, and
thus  ${\cal X}$ is partitioned into a disjoint union of $\langle\mu\rangle$-orbits (abbreviation: $\mu$-orbits).
The following result appeared previously in \cite[Lemma 3.1]{Liu}.

\begin{lem}\label{adic-orbit}
Let $\mu$ be a permutation of a finite set ${\cal X}$ and let $m$ be a positive integer.
Then the following statements are equivalent:
\begin{itemize}
\item[\bf(i)] There is a partition ${\cal X}={\cal X}_0\cup {\cal X}_1\cup\cdots\cup
{\cal X}_{m-1}$ such that $\mu({\cal X}_i)={\cal X}_{i+1}$ for
$i=0,1,\cdots,m-1$ (the subscripts are taken modulo $m$).
\item[\bf(ii)]The length of every $\mu$-orbit on ${\cal X}$ is divisible by $m$.
\end{itemize}
\end{lem}

Let a finite group $G$ act on a finite set ${\cal X}$. As is well
known, for $x\in{\cal X}$, the length of the $G$-orbit containing
$x$ is equal to the index $|G:G_x|$, where $G_x$ is the stabilizer
of $x$ in $G$. The action of $G$ on ${\cal X}$ is said to be {\em
free} if, for any $x\in{\cal X}$, the stabilizer of $x$ is $G_x= \{
1 \}$. An element  $\mu\in G$ is said to be free on ${\cal X}$ if
the subgroup $\langle \mu\rangle$ generated by $\mu$ acts on ${\cal
X}$ freely; in that case the length of any $\mu$-orbit on ${\cal X}$
is equal to the order of $\mu$   (cf. \cite[Ch.1]{AB}).

The proofs of the next two elementary facts are straightforward,
so we omit them here.

\begin{lem}\label{lcm orbit}
Let $G$, $H$ be finite groups,
and let ${\cal X}$, ${\cal Y}$ be a finite $G$-set and a finite $H$-set, respectively.
Then ${\cal X}\times {\cal Y}$ is a finite $(G\times H)$-set with the
natural action of $G\times H$, and the following statements hold:

\begin{itemize}
\item[\bf(i)]
For $g\in G$ and $h\in H$, the order of
$(g,h)\in G\times H$ is equal to the least common multiple of
the order of $g$ in $G$ and the order of $h$ in $H$, i.e.,
${\rm lcm}\big({\rm ord}_G(g),{\rm ord}_H(h)\big)$.

\item[\bf(ii)]
For $x\in{\cal X}$ and $y\in{\cal Y}$, the length
of the $(g,h)$-orbit on ${\cal X}\times {\cal Y}$ containing $(x,y)$
is equal to the least common multiple of
the length of the $g$-orbit on ${\cal X}$ containing $x$ and
the length of the $h$-orbit on ${\cal Y}$ containing $y$.
\end{itemize}
\end{lem}

\begin{lem}\label{quotient-orbit}
Let $G$ act
on a finite set ${\cal X}$ freely, and let $N$ be a normal subgroup of $G$.
Let ${\cal X}/N$ be the set of $N$-orbits on ${\cal X}$.
Then the quotient $G/N$ acts on ${\cal X}/N$ freely;
in particular, the length of any $G/N$-orbit on ${\cal X}/N$ is equal
to the index $|G:N|$.
\end{lem}

\begin{Remark}\label{modulo p}\rm
Let $t$ be a positive integer and $u\in{\Bbb Z}_t^*$.
\begin{itemize}

\item[\bf(i)]
The action of $\mu_u$ on ${\Bbb Z}_t$ is not free,
e.g., $0$ is always fixed by~$\mu_u$.
However, ${\Bbb Z}_t^*$ is $\mu_u$-invariant and the action of $\mu_u$
on ${\Bbb Z}_t^*$ is always free.

\item[\bf(ii)]
If $t=p^a$ is an odd prime power, then
${\Bbb Z}_{p^a}^*$ is a cyclic group of order $p^{a-1}(p-1)$;
the subset $1+p^b{\Bbb Z}_{p^a}$ with $b\ge 1$ of ${\Bbb Z}_{p^a}^*$
is a cyclic subgroup of order $p^{\max\{a-b,\,0\}}$,
and $1+p^bd$, with $d$ coprime to $p$,
is a generator of the cyclic subgroup $1+p^b{\Bbb Z}_{p^a}$.

\item[\bf(iii)]
If $t=2^a$ with $a\ge 2$, then
${\Bbb Z}_{2^a}^*=\begin{cases}\langle-1\rangle, & a=2;\\
  \langle-1\rangle\times\langle 5\rangle, & a>2,\end{cases}$
where the order of $\langle 5 \rangle$ is  $|\langle 5\rangle|=2^{a-2}$.
For the subgroup $1+2^b{\Bbb Z}_{2^a}$ with $b\ge 1$ of ${\Bbb Z}_{2^a}^*$,
there are two subcases:
 \begin{itemize}

  \item[\bf(iii.1)]
   if $b\ge 2$, then
   $1+2^b{\Bbb Z}_{2^a}\subseteq\langle 5\rangle$
   with order $|1+2^b{\Bbb Z}_{2^a}|=2^{\max\{a-b,\,0\}}$,
   and $1+2^bd$, with $d$ coprime to $2$,
is a generator of the cyclic subgroup $1+2^b{\Bbb Z}_{2^a}$.

  \item[\bf(iii.2)] if $b=1$, then $1+2{\Bbb Z}_{2^a}={\Bbb Z}_{2^a}^*$.

 \end{itemize}
\end{itemize}
\end{Remark}

The next two lemmas play important roles in the proofs of our main results.

\begin{lem}\label{modulo 2}
Let $u\ne -1$ be an odd integer. In the multiplicative group
${\Bbb Z}_{2^a}^*$ ($a\ge 2$), we have:
\begin{itemize}

\item[\bf(i)]
If $\nu_2(u-1)\ge 2$, then
$\langle u\rangle\subseteq\langle 5\rangle$,
${\rm ord}(u)=2^{\max\{a-\nu_2(u-1),\,0\}}$ and the quotient group
${\Bbb Z}_{2^a}^*/\langle u\rangle
=\langle\overline{-1}\rangle\times\langle\bar 5\rangle$, where
$\overline{-1}$ and $\bar 5$ denote the images of $-1$ and $5$
in the quotient group,  respectively.
In particular,  $|\langle\bar 5\rangle|=2^{\min\{\nu_2(u-1)-2,\,a-2\}}$ and
$|\langle\overline{-1}\rangle|=2$.

\item[\bf(ii)]
If $\nu_2(u-1)=1$, then $\langle u\rangle\cap\langle 5\rangle=\langle u^2\rangle$,
${\rm ord}(u^2)=2^{\max\{a-\nu_2(u+1)-1,\,0\}}$ and the quotient group
${\Bbb Z}_{2^a}^*/\langle u\rangle=\langle\bar 5\rangle$
is a cyclic group of order
$2^{\min\{\nu_2(u+1)-1,\,a-2\}}$.

\end{itemize}
\end{lem}
\begin{proof}
Let
$u=1+2^bd$,  where $b=\nu_2(u-1)$ and $d$ is odd.

(i).~ If $b\ge a$, then $u=1\in\langle 5\rangle$.
Otherwise, $2\le b<a$,  so $(1+2^bd)^{2^{a-b}}=1$ but
$(1+2^bd)^{2^{a-b-1}}\ne 1$, which gives ${\rm ord}(u)=2^{a-b}$.
In other words, ${\rm ord}(u)=2^{\max\{a-b,\,0\}}$
\big(this is just an argument for Remark~\ref{modulo p} (iii.1)\big).
Specifically, $\langle 5\rangle=1+2^2{\Bbb Z}_{2^a}$.
Therefore, $\langle u\rangle\subseteq\langle 5\rangle$,
${\Bbb Z}_{2^a}^*/\langle u\rangle
=\langle\bar 5\rangle\times\langle\overline{-1}\rangle$ and
$$|\langle\bar 5\rangle|=2^{a-2}/{\rm ord}(u)=2^{\min\{b-2,\,a-2\}}.$$

(ii).~ In this case  $u+1=2(d+1)$, so $\nu_2(u+1)=1+\nu_2(d+1)\ge 2$.
Writing $d=2^{\nu_2(u+1)-1}(-u')-1$ with $u'$ being odd, we get
\begin{equation}\label{2-valuation=1}
u=1+2d=(-1)(1+2^{\nu_2(u+1)}u')
 \in\langle -1\rangle\times\langle 5\rangle,\quad
  \nu_2(u+1)\ge 2,~~ 2\nmid u'.
\end{equation}
Note that  $u\notin\langle 5\rangle$, but
$u^2=1+2^{\nu_2(u+1)+1}(u'+2^{\nu_2(u+1)-1}u'^2)\in\langle 5\rangle$ and
$$
{\rm ord}(u^2)=\begin{cases}2^{a-\nu_2(u+1)-1}, & \nu_2(u+1)+1\le a ; \\
  1, & \nu_2(u+1)+1> a. \end{cases}
$$
In other words, ${\rm ord}(u^2)=2^{\max\{a-\nu_2(u+1)-1,\,0\}}$.
Then ${\Bbb Z}_{2^a}^*=\langle 5\rangle\cdot\langle u\rangle$
and $\langle u\rangle\cap\langle 5\rangle=\langle u^2\rangle$, hence
$$
{\Bbb Z}_{2^a}^*/\langle u\rangle\cong
\langle 5\rangle/ \left( \langle u\rangle\cap\langle 5\rangle \right)
=\langle 5\rangle/\langle u^2\rangle
$$
is cyclic group. Recalling that $|\langle 5\rangle|=2^{a-2}$, we then have
$$
|{\Bbb Z}_{2^a}^*/\langle u\rangle|=|\langle 5\rangle/\langle u^2\rangle|
=2^{\min\{\nu_2(u+1)-1,\,a-2\}}.
$$
We are done.
\end{proof}

\begin{lem}\label{modolo 2'}
Let $h,u$ be odd integers with $u\ne -1$.
We denote by $\bar h$ the image of $h$ in the quotient group
${\Bbb Z}_{2^a}^*/\langle u\rangle$, where $a\ge 2$.
Let ${\rm ord}(\bar h)=2^{v}$.
With the convention that $\nu_2(h+1)=-\infty$ and
$|\nu_2(h+1)|=\infty$ when $h=-1$, we have:
\begin{itemize}
\item[\bf(i)]
If both $\nu_2(u-1)\ge 2$ and $\nu_2(h-1)\ge 2$, then
 $$ v=
  \max\big\{\min\{\nu_2(u-1),\,a\}-\nu_2(h-1),\,0\big\}.$$
\item[\bf(ii)]
If $\nu_2(u-1)=1$ and $\nu_2(h-1)\ge 2$, then
 $$v=
   \max\big\{\min\{\nu_2(u+1)+1,\,a\}-\nu_2(h-1),\,0\big\}.$$
\item[\bf(iii)] If $\nu_2(u-1)\ge 2$ and $\nu_2(h-1)=1$, then
 $$v=
 \max\big\{\min\{\nu_2(u-1),\,a\}-|\nu_2(h+1)|,\,1\big\}.$$
\item[\bf(iv)]
If both $\nu_2(u-1)=1$ and $\nu_2(h-1)=1$, then
$$v=\left\{\begin{array}{r}\max\big\{
 \min\{\nu_2(u+1)+1,\,a\}-\min\{|\nu_2(h+1)|,\,\nu_2(u+1)\},\,0\big\},\quad \\
  {\rm if}~~ \nu_2(h+1)\ne\nu_2(u+1);\\[2mm]
  0, \hfill{\rm if}~~ \nu_2(h+1)=\nu_2(u+1).
\end{array}\right.$$
\end{itemize}
\end{lem}

\begin{proof} Let $\langle h,u\rangle$ be the subgroup of
${\Bbb Z}_{2^a}^*$ generated by $h$ and $u$. Then
$2^v=\big|\langle h,u\rangle/\langle u\rangle\big|$.

(i).~ By Lemma~\ref{modulo 2}(i), both $u$ and $h$ are located
in $\langle 5\rangle$.
Hence, $2^v={\rm ord}(h)/{\rm ord}(u)$ if ${\rm ord}(h)>{\rm ord}(u)$,
and $2^v=1$ otherwise. However, ${\rm ord}(h)=2^{\max\{a-\nu_2(h-1),\,0\}}$
and ${\rm ord}(u)=2^{\max\{a-\nu_2(u-1),\,0\}}$
by Lemma~\ref{modulo 2}(i) again.  We get
$${\rm ord}(h)>{\rm ord}(u)~~ \iff~~
       \nu_2(h-1)<\min\{\nu_2(u-1),a\}.$$
Thus $2^v=2^{\max\{\min\{\nu_2(u-1),\,a\}-\nu_2(h-1),\,0\}}$.

(ii).~ In this case, $\langle h\rangle\subseteq \langle 5\rangle$
but $\langle u\rangle\cap\langle 5\rangle=\langle u^2\rangle$
(see Lemma~\ref{modulo 2}(ii)).  Then
$\langle h\rangle\cap\langle u\rangle=
\langle h\rangle\cap\langle 5\rangle\cap\langle u\rangle=
\langle h\rangle\cap\langle u^2\rangle$,
and hence
$$
2^v=\big|\langle h,u\rangle / \langle u\rangle\big|
=\big|\langle h\rangle / \langle h\rangle\cap\langle u^2\rangle\big|,
$$
which leads to computations in the cyclic group $\langle 5\rangle$
similar to the ones in (i).
By Lemma~\ref{modulo 2}(ii),
${\rm ord}(u^2)=2^{\max\{a-\nu_2(u+1)-1,\,0\}}$, so
$$2^v=2^{\max\{\min\{\nu_2(u+1)+1,\,a\}-\nu_2(h-1),\,0\}}.$$

(iii).~ We can write $h=(-1)(1+2^{\nu_2(h+1)}h')$
with $|\nu_2(h+1)|\ge 2$
(the case $\nu_2(h+1)=-\infty$, i.e., $h=-1$, is allowed)
and $h'$ being odd  (see Eqn~(\ref{2-valuation=1}));
set $h''=1+2^{\nu_2(h+1)}h'$.
By Lemma~\ref{modulo 2}(i), it follows that
${\Bbb Z}_{2^a}^*/\langle u\rangle=
\langle\overline{-1}\rangle\times\langle\bar 5\rangle$,
and $\bar h=\overline{-1}\cdot\overline{h''}$ with
$\overline{h''}\in\langle\bar 5\rangle$. Then
$$
 2^v={\rm ord}(\bar h)=
 \max\{{\rm ord}(\overline{-1}),\,{\rm ord}(\overline{h''})\} .
$$
We know that ${\rm ord}(h'')=2^{\max\{a-|\nu_2(h+1)|,\,0\}}$
(see Lemma~\ref{modulo 2}(i), but $h=-1$ is allowed here).
With the same argument as in (i), we have
$$
{\rm ord}(\overline{h''})=2^{\max\{\min\{\nu_2(u-1),\,a\}-|\nu_2(h+1)|,\,0\}}.
$$
Recalling that ${\rm ord}(\overline{-1})=2$, we get that
$
2^v=2^{\max\{\min\{\nu_2(u-1),\,a\}-|\nu_2(h+1)|,\,1\}}.
$

(iv).~ Similar to the above, we can write
\begin{eqnarray*}
 & h=(-1)h'',~~ h''=1+2^{\nu_2(h+1)}h',~~ \nu_2(h+1)\ge 2,~~ 2\nmid h';\\
 & u=(-1)u'',~~ u''=1+2^{\nu_2(u+1)}u',~~ \nu_2(u+1)\ge 2,~~ 2\nmid u'.
\end{eqnarray*}
It is clear that $\langle h,u\rangle=\langle hu,u\rangle$, so
$$
2^v=|\langle h,u\rangle/\langle u\rangle|=
  |\langle hu,u\rangle/\langle u\rangle|.
$$
Since $u\notin\langle 5\rangle$, but
$u^2=1+2^{\nu_2(u+1)+1}(u'+2^{\nu_2(u+1)-1}u'^2)\in\langle 5\rangle$ and
$$
 hu=h''u''=1+2^{\nu_2(h+1)}h'+2^{\nu_2(u+1)}u'+2^{\nu_2(h+1)+\nu(u+1)}h'u'
 ~\in~\langle 5\rangle;
$$
this present case can be reduced to the case (ii) above (by replacing $h$ in (ii) with $hu$).
If $\nu_2(h+1)=\nu_2(u+1)$, then $\nu_2(hu-1)\ge \nu_2(u+1)+1$
(see the above formulation of $hu$), hence
$2^v=1$ by the conclusion in (ii).
Otherwise
$$\nu_2(hu-1)=\begin{cases}
 \min\{\nu_2(h+1),\nu_2(u+1)\}, & h\ne -1;\\
 \nu_2(u+1), & h=-1, \end{cases}$$
i.e., $\nu_2(hu-1)=\min\{|\nu_2(h+1)|,\nu_2(u+1)\}$.
By the conclusion in (ii) again,
$$2^v= 2^{\max\left\{
  \min\{\nu_2(u+1)+1,\,a\}-\min\{|\nu_2(h+1)|,\,\nu_2(u+1)\},\,0\right\}}.
$$
If $\nu_2(h+1)>\nu_2(u+1)$ or $h=-1$, then
$\min\{|\nu_2(h+1)|,\,\nu_2(u+1)\}=\nu_2(u+1)$,
and hence $v$ can be simplified to:
$$v={\max\{\min\{\nu_2(u+1)+1,a\}-\nu_2(u+1),\,0\}}
  =\begin{cases}1, & \nu_2(u+1)<a;\\ 0, & \nu_2(u+1)\ge a.  \end{cases}$$
The proof is completed.
\end{proof}

\section{Proofs of the main results}

We keep the notations of Section 2.
Consider the surjective homomorphism
\begin{equation}\label{rho}
{\Bbb Z}_{rn} \longrightarrow {\Bbb Z}_r,\quad
x~({\rm mod}~rn) \longmapsto x~({\rm mod}~r).
\end{equation}
Then $1+r{\Bbb Z}_{rn}$ is just the inverse image of $1\in{\Bbb Z}_r$.

Assume that
$p_1,\cdots,p_k$, $p'_1,\cdots,p'_{k'}$, $p''_1,\cdots,p''_{k''}$
are distinct primes such that
$$
n=p_1^{\alpha_1}\cdots p_k^{\alpha_k}
  {p'_1}^{\alpha'_1}\cdots{p'_{k'}}^{\alpha'_{k'}},\quad
  \mbox{with $\alpha_1, \cdots, \alpha_k, \alpha'_1, \cdots, \alpha'_{k'}$
  all positive;}
$$
$$
r=p_1^{\beta_1}\cdots p_k^{\beta_k}
  {p''_1}^{\beta''_1}\cdots{p''_{k''}}^{\beta''_{k''}}, \quad
  \mbox{with $\beta_1, \cdots, \beta_k, \beta''_1, \cdots, \beta''_{k''}$
  all positive,}
$$
i.e., $\alpha_i=\nu_{p_i}(n)$, $\beta_i=\nu_{p_i}(r)$, etc. Then
$$
rn=p_1^{\alpha_1+\beta_1}\cdots p_k^{\alpha_k+\beta_k}
   {p'_1}^{\alpha'_1}\cdots{p'_{k'}}^{\alpha'_{k'}}
   {p''_1}^{\beta''_1}\cdots{p''_{k''}}^{\beta''_{k''}}.
$$
Set $n'={p'_1}^{\alpha'_1}\cdots{p'_{k'}}^{\alpha'_{k'}}$
and $r''={p''_1}^{\beta''_1}\cdots{p''_{k''}}^{\beta''_{k''}}$.
Applying the Chinese Remainder Theorem, we rewrite ${\Bbb Z}_{rn}$ as follows:
$$
{\Bbb Z}_{rn}\mathop{=}^{\rm CRT}
{\Bbb Z}_{p_1^{\alpha_1+\beta_1}}\times\cdots\times
 {\Bbb Z}_{p_k^{\alpha_k+\beta_k}} \times{\Bbb Z}_{n'}\times{\Bbb Z}_{r''}.
$$
The surjective homomorphism (\ref{rho}) may be rewritten as
\begin{equation}\label{rho CRT} \rho:\quad
{\Bbb Z}_{p_1^{\alpha_1+\beta_1}}\times\cdots\times
 {\Bbb Z}_{p_k^{\alpha_k+\beta_k}} \times{\Bbb Z}_{n'}\times{\Bbb Z}_{r''}
 ~{\longrightarrow}~
 {\Bbb Z}_{p_1^{\beta_1}}\times\cdots\times
 {\Bbb Z}_{p_k^{\beta_k}}\times{\Bbb Z}_{r''} ,
\end{equation}
with kernel
$$
{\rm Ker}(\rho)=p_1^{\beta_1}{\Bbb Z}_{p_1^{\alpha_1+\beta_1}}
  \times\cdots\times
 p_k^{\beta_k}{\Bbb Z}_{p_k^{\alpha_k+\beta_k}}
 \times{\Bbb Z}_{n'}\times\{0\},
$$
where $\{0\}$ is the zero ideal of ${\Bbb Z}_{r''}$.
Thus, $1+r{\Bbb Z}_{rn}=1+{\rm Ker}(\rho)$ can be  rewritten as
\begin{equation}\label{Omega CRT}
1+r{\Bbb Z}_{rn}\mathop{=}^{\rm CRT}
 \big(1+p_1^{\beta_1}{\Bbb Z}_{p_1^{\alpha_1+\beta_1}}\big)
  \times\cdots\times
  \big(1+p_k^{\beta_k}{\Bbb Z}_{p_k^{\alpha_k+\beta_k}}\big)
 \times{\Bbb Z}_{n'}\times\{1\}.
\end{equation}
Therefore
\begin{equation}\label{Z Omega CRT}
{\Bbb Z}_{rn}^*\cap(1+r{\Bbb Z}_{rn})\mathop{=}^{\rm CRT}
 \big(1+p_1^{\beta_1}{\Bbb Z}_{p_1^{\alpha_1+\beta_1}}\big)
  \times\cdots\times
  \big(1+p_k^{\beta_k}{\Bbb Z}_{p_k^{\alpha_k+\beta_k}}\big)
 \times{\Bbb Z}_{n'}^*\times\{1\}.
\end{equation}
Thus, any $x\in(1+r{\Bbb Z}_{rn})$ can be represented as
\begin{equation}\label{x CRT}
 x\mathop{=}^{\rm CRT}\big( 1+ p_1^{\xi_1}x_1,~\cdots,~
 1+ p_k^{\xi_k}x_k,~ x',~ 1\big)
\end{equation}
with $\xi_i=\nu_{p_i}(x-1)\ge\beta_i$,  $p_i\nmid x_i$ for $i=1,\cdots,k$,
and $x'\in{\Bbb Z}_{n'}$, and hence
any $s\in{\Bbb Z_{rn}}^*\cap(1+r{\Bbb Z}_{rn})$ can be represented as
\begin{equation}\label{s CRT}
 s\mathop{=}^{\rm CRT}\big( 1+ p_1^{\sigma_1}s_1,~\cdots,~
 1+ p_k^{\sigma_k}s_k,~ s',~ 1\big)
\end{equation}
with $\sigma_i=\nu_{p_i}(s-1)\ge\beta_i$ and $p_i\nmid s_i$ for
$i=1,\cdots,k$, and $s'\in{\Bbb Z}_{n'}^*$. In particular, since
$q\in{\Bbb Z}_{rn}^*\cap(1+r{\Bbb Z}_{rn})$, we have
\begin{equation}\label{q CRT}
 q \mathop{=}^{\rm CRT} \big( 1+ p_1^{\tau_1}q_1,~\cdots,~
 1+ p_k^{\tau_k}q_k,~ q',~ 1\big)
\end{equation}
where $\tau_i=\nu_{p_i}(q-1)\ge\beta_i$ and
$p_i\nmid q_i$ for $i=1,\cdots,k$, and $q'\in{\Bbb Z}_{n'}^*$.

We first need the following observation.

\begin{lem}\label{outside p}
Let $p$ be a prime.
If $p\notin\{p_1,\cdots,p_k\}$, then there is a $\mu_s$-orbit
in $1+r{\Bbb Z}_{rn}$  whose length is not divisible by $p$.
\end{lem}

\begin{proof}
Take $x_0\in(1+r{\Bbb Z}_{rn})$ such that (cf. Eqn (\ref{x CRT})):
$$
 x_0\mathop{=}^{\rm CRT}\big( 1+ p_1^{\xi_1}x_1,~\cdots,~
 1+ p_k^{\xi_k}x_k,~ 0 ,~ 1\big).
$$
Then the length of the $\mu_s$-orbit in ${\Bbb Z}_{n'}$ containing $0$
is $1$, and the length of the $\mu_s$-orbit in
$1+p_i^{\beta_i}{\Bbb Z}_{p_i^{\alpha_i+\beta_i}}$
containing $1+ p_i^{\xi_i}x_i$ is a power of $p_i$.
By Lemmas~\ref{lcm orbit} and \ref{quotient-orbit}, we see that
the length of the $\mu_s$-orbit in $(1+r{\Bbb Z}_{rn})/\mu_q$
containing the $q$-cyclotomic coset $x_0\langle q\rangle$
is not divisible  by $p$.
\end{proof}

\medskip
Recall that, for any $p_i$, $1\le i\le k$, both
$q$ and $s\pmod{p_i^{\alpha_i+\beta_i}}$ are contained in
the multiplicative group $1+p_i^{\beta_i}{\Bbb Z}_{p_i^{\alpha_i+\beta_i}}$;
we denote by $\langle q,s\rangle_i$ the subgroup of
$1+p_i^{\beta_i}{\Bbb Z}_{p_i^{\alpha_i+\beta_i}}$
generated by $q$ and $s$.

\begin{Theorem}\label{M_s by quotient}
Let
$M_s=\prod\limits_{i=1}^k |\langle q,s\rangle_i:\langle q\rangle_i|$
be the order of $s$ in the quotient group
$\prod\limits_{i=1}^k\big(1+p_i^{\beta_i}{\Bbb Z}_{p_i^{\alpha_i+\beta_i}}\big)
 /\langle q\rangle_i$.
Then $\mu_s$ is a Type I $m$-adic splitting for $1+r{\Bbb Z}_{rn}$
if and only if $m|M_s$.
\end{Theorem}

\begin{proof}
For any $x\in(1+r{\Bbb Z}_{rn})$ as in Eqn (\ref{x CRT}),
by Lemma \ref{quotient-orbit},
the length of the $\mu_s$-orbit in the quotient set
$\big(1+p_i^{\beta_i}{\Bbb Z}_{p_i^{\alpha_i+\beta_i}}\big)/\mu_q$
containing $(1+ p_i^{\xi_i}x_i)\langle q\rangle$ is equal to
$|\langle q,s\rangle_i:\langle q\rangle_i|$.
By Lemmas~\ref{adic-orbit},
\ref{lcm orbit} and \ref{outside p},
the theorem follows at once.
\end{proof}

Now we are ready to prove our main results.

\medskip
\noindent {\it Proof of Theorem~\ref{theorem1}.}~
Take an $\hat s\in{\Bbb Z}_{rn}^*\cap(1+r{\Bbb Z}_{rn})$ as follows
\big(cf. Eqn~(\ref{s CRT})\big):
$$
\hat s\mathop{=}^{\rm CRT}
 (1+p_1^{\hat \sigma_1},\,\cdots,\,1+p_k^{\hat\sigma_k},\,1,1)
$$
such that each component
$1+p_i^{\hat\sigma_i}\in 1+p_i^{\beta_i}{\Bbb Z}_{p_i^{\alpha_i+\beta_i}}$
of $\hat s$ becomes an element of maximal order in the quotient group
$\big(1+p_i^{\beta_i}{\Bbb Z}_{p_i^{\alpha_i+\beta_i}}\big)
 /\langle q\rangle_i$ for $i=1,\cdots,k$.
Set $M=M_{\hat s}$ as in Theorem~\ref{M_s by quotient}. Then,
for any $s\in {\Bbb Z}_{rn}^*\cap(1+r{\Bbb Z}_{rn})$, by
Theorem~\ref{M_s by quotient} we have $M_s\,|\,M$. Thus, for an
integer $m$, an $m$-adic $(q,n,r)$-constacyclic code of Type I
exists if and only if $m$ is a divisor of~$M$. It remains to
determine $M$ by its $p$-adic valuations $\nu_p(M)$, for all primes
$p$. If $p\notin\{p_1,\cdots,p_k\}$, then we have seen from
Theorem~\ref{M_s by quotient} that $\nu_p(M)=0$. For $1\le i\le k$,
by Theorem~\ref{M_s by quotient} and the choice of $\hat s$, we see
that $p_i^{\nu_{p_i}(M)}$ is the maximal order of elements of the
quotient group $\big(1+p_i^{\beta_i}{\Bbb
Z}_{p_i^{\alpha_i+\beta_i}}\big)
 /\langle q\rangle_i$; we determine it in the following cases.

{\it Case 1:} $p_i$ is odd or $\beta_i=\nu_{p_i}(r)\ge 2$. Then the group
$1+p_i^{\beta_i}{\Bbb Z}_{p_i^{\alpha_i+\beta_i}}$
is a cyclic group of order $p_i^{\alpha_i}$, and
the order of $q$ is
$p_i^{\max\{\alpha_i+\beta_i-\nu_{p_i}(q-1),\,0\}}$
\big(recall that $\nu_{p_i}(q-1)\ge \beta_i$ and the order of $q$
is $1$ when $\alpha_i+\beta_i\le\nu_{p_i}(q-1)$\big).
Thus, the maximal order of elements of the quotient group
$\big(1+p_i^{\beta_i}{\Bbb Z}_{p_i^{\alpha_i+\beta_i}}\big)
 /\langle q\rangle_i$
is $p_i^{\min\{\nu_{p_i}(q-1)-\beta_i,\,\alpha_i\}}$;
hence $\nu_{p_i}(M)=\min\{\nu_{p_i}(q-1)-\nu_{p_i}(r),\,\nu_{p_i}(n)\}$.

{\it Case 2:} $p_i=2$ and $\nu_{2}(r)=1$.
Then the group
$1+2{\Bbb Z}_{2^{\alpha_i+1}}={\Bbb Z}_{2^{\alpha_i+1}}^*
=\langle-1\rangle\times\langle 5\rangle$
and $|\langle 5\rangle|=2^{\alpha_i-1}$. There are two subcases:

{\it Subcase 2.1:} $\nu_2(q-1)\ge 2$. By Lemma~\ref{modulo 2}(i),
the quotient group ${\Bbb Z}_{2^{\alpha_i+1}}^*/\langle q\rangle$
is a direct product of a cyclic group of order
$2^{\min\{\nu_2(q-1)-2,\,\alpha_i-1\}}$ and a group of order $2$.
Thus
$$\nu_{2}(M)=\max\{\min\{\nu_2(q-1)-2,\,\nu_{2}(n)-1\},\,1\}.$$

{\it Subcase 2.2:} $\nu_2(q-1)=1$.
By Lemma~\ref{modulo 2}(ii),
the quotient group ${\Bbb Z}_{2^{\alpha_i+1}}^*/\langle q\rangle$
is a cyclic group of order
$2^{\min\{\nu_2(q+1)-1,\,\alpha_i-1\}}$.
Thus
$$\nu_{2}(M)=\min\{\nu_2(q+1)-1,\,\nu_{2}(n)-1\}.\eqno\qed$$

\medskip
\noindent {\it Proof of Theorem~\ref{theorem2}}.~
For $i=1,\cdots,k$, from Theorem~\ref{M_s by quotient},
in the quotient group
$\big(1+p_i^{\beta_i}{\Bbb Z}_{p_i^{\alpha_i+\beta_i}}\big)
 /\langle q\rangle_i$, we have seen that
$$
\nu_{p_i}(M_s)=|\langle q,s\rangle_i:\langle q\rangle_i|.
$$
If $p_i$ is odd, then
$1+p_i^{\beta_i}{\Bbb Z}_{p_i^{\alpha_i+\beta_i}}$ is a cyclic
group of order $p_i^{\nu_{p_i}(rn)-\nu_{p_i}(r)}$.
By Remark~\ref{modulo p}(ii), we get at once that
$|\langle q,s\rangle_i:\langle q\rangle_i|
 =p_i^{\max\{\min\{\nu_{p_i}(q-1),\,\nu_{p_i}(rn)\}-|\nu_{p_i}(s-1)|,\,0\}}$.
Otherwise, $p_i=2$ and all of the conclusions follow from
Lemma~\ref{modolo 2'} immediately.
\qed

\section{Corollaries and examples}

Most results on the existence of Type I polyadic constacyclic codes can
follow as consequences from the main theorems immediately.
Moreover,
with the help of the main results and the arguments,
some interesting examples can be constructed from
Type~{\rm I} polyadic constacyclic codes.
Here we describe the case when $m=p$ is a prime,
which is an interesting case.

\subsection{$p$-adic constacyclic codes}

\begin{Corollary}\label{cor-prime}
Let $m=p$ be a prime.
Then $p$-adic $(q,n,r)$-constacyclic codes of Type I exist if and only if
one of the following two conditions holds:
\begin{itemize}
\item[\bf(i)]
 $\nu_p(n)\ge 1$ and $\nu_p(q-1)>\nu_p(r)\ge 1$
 (the case $p=2$ is allowed);
\item[\bf(ii)] $p=2$,
$\nu_2(r)=1$ and $\min\{\nu_2(q+1),\nu_2(n)\}\ge 2$.
\end{itemize}
\end{Corollary}

\begin{proof}
Taking  $m=p$ in Theorem~\ref{theorem1}, we obtain the desired result.
\end{proof}

\begin{Remark}\rm
The case of $p=2$ in Corollary~\ref{cor-prime}, i.e.,
the necessary and sufficient conditions for the existence
of duadic constacyclic codes,
was treated in \cite[Corollary 17]{Blackford13} and stated
in different notations.

On the other hand,
if the prime $p$ is odd, then (ii) of Corollary~\ref{cor-prime}
is not applicable, hence the statement can be shortened; for example,
for $p=3$, the statement can read as

``{\it Triadic $(q,n,r)$-constacyclic codes of Type I exist if and only if
$\nu_3(n)\ge 1$ and $\nu_3(q-1)>\nu_3(r)\ge 1$.}''

\noindent
This result has been obtained in~\cite{Liu}.

\end{Remark}

Inspired by the conditions of Corollary~\ref{cor-prime}, we construct
a class of $p$-adic constacyclic generalized Reed-Solomon codes.
For nonzero $v_0,v_1,\cdots,v_{n-1}\in {\Bbb F}_q^*$ and
distinct $\alpha_0,\alpha_1,\cdots,\alpha_{n-1}\in {\Bbb F}_q$,
the following $[n,k,n-k+1]$ code
$$
 \Big\{\Big(v_0f(\alpha_0),v_1f(\alpha_1),\cdots,
  v_{n-1}f(\alpha_{n-1})\Big)\,\Big|\,
  f(X)\in{\Bbb F}_q[X],~\deg f(X)<k\Big\}
$$
is called a {\em generalized Reed-Solomon code},
abbreviated by {\em GRS code}, with locator
$\mbox{\boldmath$\alpha$}=(\alpha_0,\alpha_1,\cdots,\alpha_{n-1})$;
we denote this GRS code by
${\rm GRS}_k(\mbox{\boldmath$\alpha$};{\mathbf v})$,
where ${\mathbf v}=(v_0,v_1,\cdots,v_{n-1})$
(cf. \cite[Ch.9]{Lingbook}).

\begin{Proposition}\label{p-adic}
Assume that $m=p$ is a prime, $q$ is a prime power with
$\nu_p(q-1)\ge 2$, and $rn\mid(q-1)$ such that
$\nu_p(r)\geq 1$ and $\nu_p(n)\geq 1$
(then Corollary \ref{cor-prime}(i) is satisfied).
Let $\omega\in{\Bbb F}_q$ be a primitive $rn${\rm th} root of unity
and $\lambda=\omega^n$. Set
$$\textstyle
{\cal X}_j=\Big\{1+ir\,\Big|\,\frac{jn}{p}\leq i<\frac{(j+1)n}{p}\Big\},\qquad
 j=0,1,\cdots,p-1.
$$
Then
\begin{itemize}

\item[{\bf(i)}] $C_{{\cal X}_j}$, for $j=0,1,\cdots,p-1$, are
Type I $p$-adic $\lambda$-constacyclic codes
given by $\mu_{1+\frac{rn}{p}}$;

\item[{\bf(ii)}] for any $0<k<p$, the constacyclic code
$C=C_{{\cal X}_0}\oplus C_{{\cal X}_1}\oplus\cdots\oplus C_{{\cal X}_{k-1}}$
is the $[n,\frac{kn}{p},\frac{(p-k)n}{p}+1]$ GRS code
${\rm GRS}_{kn/p}(\mbox{\boldmath$\omega$};{\mathbf v})$,
where $\mbox{\boldmath$\omega$}=(1,\omega^{-r},\cdots,\omega^{-(n-1)r})$
and ${\mathbf v}=(1,\omega^{-1},\cdots,\omega^{-(n-1)})$.
\end{itemize}
\end{Proposition}

\begin{proof} (i).~
Let $s=1+\frac{rn}{p}$. Noting that $p\mid r$, we have
$$
\mu_s(1+ir)=\Big(1+\frac{rn}{p}\Big)(1+ir)\equiv~
1+r\Big(\frac{n}{p}+i\Big)~(\bmod~rn).
$$
Hence $\mu_s({\cal X}_j)={\cal X}_{j+1}$ for $j=0,\cdots,p-2$,
and $\mu_s({\cal X}_{p-1})={\cal X}_0$.
Thus ${\cal X}_0$, ${\cal X}_1$, $\cdots$, ${\cal X}_{p-1}$ form a
Type I $p$-adic splitting of $1+r{\Bbb Z}_{rn}$ given by $\mu_s$.

(ii).~ Since $\omega^{-r}$ is a primitive $n$th root of unity,
$\mbox{\boldmath$\omega$}$ is a locator, and
$$
{\rm GRS}_{kn/p}(\mbox{\boldmath$\omega$};{\mathbf v})
=\Big\{\big(f(1),\omega^{-1}f(\omega^{-r}),\cdots,
  \omega^{-(n-1)}f(\omega^{-r(n-1)})\big)\,\Big|\,
  f(X)\in{\Bbb F}_q[X],~\deg f(X)<\frac{kn}{p}\Big\}
$$
is a GRS $[n,\frac{kn}{p},\frac{(p-k)n}{p}+1]$ code.
We need to show that
$C={\rm GRS}_{kn/p}(\mbox{\boldmath$\omega$};{\mathbf v})$.
Since $\dim C=\frac{kn}{p}$, it suffices to show that
${\rm GRS}_{kn/p}(\mbox{\boldmath$\omega$};{\mathbf v})\subseteq C$.

Set ${\cal K}={\cal X}_0\cup\cdots\cup{\cal X}_{k-1}$
and ${\cal K}'={\cal X}_k\cup\cdots\cup{\cal X}_{p-1}$.
Then $\prod_{Q\in{\cal K}/\mu_q }M_Q(X)$ is a check polynomial of~$C$,
hence $\{\omega^{t}\mid t\in{\cal K}'\}
  =\{\omega^{1+ir}\mid \frac{kn}{p}\le i<n\}$ is the set of zeros
of the code $C$.

For $f(X)=\sum_{j=0}^{\frac{kn}{p}-1}f_jX^j$ with $f_i\in{\Bbb F}_q$,
the codeword
$c'_f=\big(f(1),\omega^{-1}f(\omega^{-r}),\cdots,
  \omega^{-(n-1)}f(\omega^{-r(n-1)})\big)$ in
${\rm GRS}_{kn/p}(\mbox{\boldmath$\omega$};{\mathbf v})$
corresponds to the polynomial
$c'_f(X)=\sum_{t=0}^{n-1}\omega^{-t}f(\omega^{-rt})X^t$
 in the polynomial representation of codewords.
To prove that $c'_f\in C$, it is enough to show that
$c'_f(\omega^{1+ir})=0$ for $\frac{kn}{p}\le i<n$.
We compute $c'_f(\omega^{1+ir})$ as follows:
$$
c'_f(\omega^{1+ir})=\sum_{t=0}^{n-1}\omega^{-t}f(\omega^{-rt})\omega^{(1+ir)t}
=\sum_{t=0}^{n-1}\omega^{-t} \left(
\sum_{j=0}^{\frac{kn}{p}-1}f_j\omega^{-rtj}\right) \omega^{(1+ir)t}
=\sum_{j=0}^{\frac{kn}{p}-1}f_j\sum_{t=0}^{n-1}\omega^{r(i-j)t}.
$$
Since $\frac{kn}{p}\le i<n$ and $0\le j<\frac{kn}{p}$,
we see that $0<i-j<n$, hence $\omega^{r(i-j)}\ne 1$ as $\omega^{r}$
is a primitive $n$th root of unity. Then
$$\sum_{t=0}^{n-1}\omega^{r(i-j)t}
 =\frac{\omega^{r(i-j)n}-1}{\omega^{r(i-j)}-1}=0,\qquad {\mbox{ for }}
  \frac{kn}{p}\le i<n,~~0\le j<\frac{kn}{p}.$$
Therefore, $c'_f(\omega^{1+ir})=0$ for all $1+ir\in{\cal K}'$; we are done.
\end{proof}

Some GRS codes from Proposition~\ref{p-adic} are exhibited in Table~\ref{t1}.

\begin{table}[h]\label{t1}\begin{center}
{\footnotesize\renewcommand{\arraystretch}{1.8}
\begin{tabular}{|c|c|c|c|c|c|c|c|}\hline
No & $m$ & $q$ & $n$ & $r$ & GRS code & parameters\\ \hline
(i) & 3 & 19 & 6 & 3 &
 $\big\{(f(1), \omega^{-1}f(\omega^{-3}),\cdots,\omega^{-5}f(\omega^{-15}))
 ~\big|~f(X)\in\mathbb{F}_{19}[X],~\deg f(X)<4\big\}$  &[6,4,3] \\
(ii) & 3 & $2^6$ & 21 & 3 &
 $\big\{(f(1),\omega^{-1}f(\omega^{-3}),\cdots,\omega^{-20}f(\omega^{-60}))
 ~\big|~f(X)\in\mathbb{F}_{2^6}[X],~\deg f(X)<7\big\}$  &[21,7,15] \\
(iii) & 2 & 17 & 8  & 2 &
 $\big\{(f(1),\omega^{-1}f(\omega^{-2}), \cdots,\omega^{-7}f(\omega^{-14}))
 ~\big|~f(X)\in \mathbb{F}_{17}[X],~\deg f(X)<4\big\}$      &[8,4,5] \\
(iv) & 2 & $3^4$ & 40  & 2 &
 $\big\{(f(1),\omega^{-1}f(\omega^{-2}),\cdots,\omega^{-39}f(\omega^{-78}))
 ~\big|~f(X)\in \mathbb{F}_{3^4}[X],~\deg f(X)<20\big\}$   &[40,20,21] \\
(v) & 2 & $5^2$ & 12  & 2 &
 $\big\{(f(1),\omega^{-1}f(\omega^{-2}),\cdots,\omega^{-11}f(\omega^{-22}))
 ~\big|~f(X)\in \mathbb{F}_{5^2}[X],~\deg f(X)<6\big\}$   &[12,6,7] \\
(vi) & 2 & $7^2$ & 24  & 2 &
 $\big\{(f(1),\omega^{-1}f(\omega^{-2}),\cdots,\omega^{-23}f(\omega^{-46}))
 ~\big|~f(X)\in \mathbb{F}_{7^2}[X],~\deg f(X)<12\big\}$   &[24,12,13] \\
\hline
\end{tabular}}
\begin{caption}
{GRS codes from Type {\rm I} polyadic constacyclic codes}
\end{caption}
\end{center}\end{table}

\begin{Example}\label{ex1}\rm
An interesting particular case of Proposition~\ref{p-adic}
is as follows:
$m=r=2$, $n$ is an even divisor of $\frac{q-1}{2}$,
and the splitting of $1+2{\Bbb Z}_{2n}$ is
\begin{equation}\label{dual GRS}\textstyle
{\cal X}_0=\{1,\,3,\,\cdots,\,n-1\},\qquad
 {\cal X}_1=\{n+1,\,n+3,\,\cdots,\,2n-1\},
\end{equation}
(e.g., the codes (iii)-(vi) in Table~\ref{t1} where $n=\frac{q-1}{2}$).
It is a Type I duadic splitting of $1+2{\Bbb Z}_{2n}$
given by $\mu_{1+\frac{rn}{2}}=\mu_{n+1}$. However,
it is easy to check that ${\cal X}_0,{\cal X}_1$
also form a splitting of $1+2{\Bbb Z}_{2n}$ given by $\mu_{-1}$.
In other words,
both $C_{{\cal X}_0}$ and $C_{{\cal X}_1}$ are self-dual
duadic negacyclic GRS codes with parameters
$[n,\frac{n}{2},\frac{n}{2}+1]$.

The biggest choice of $n$ is $n=\frac{q-1}{2}$ and, in this case,
the self-dual duadic negacyclic GRS code $C_{{\cal X}_0}$
has parameters $[\frac{q-1}{2},\frac{q-1}{4},\frac{q+3}{4}]$.
\end{Example}

Before further analyzing this example, we discuss the particular
case of Theorem~\ref{theorem2} where $m=p$ is a prime.

\subsection{$p$-adic constacyclic codes given by $\mu_s$}

When $m=p$ is an odd prime in Theorem~\ref{theorem2}, the case is
easy, as shown in the following.

\begin{Corollary}\label{p odd} Assume that $m=p$ is an odd prime,
$s\in{\Bbb Z}_{rn}^*\cap(1+r{\Bbb Z}_{rn})$ and $s\ne 1$. Then
Type I $p$-adic splittings of $1+r{\Bbb Z}_{rn}$ given by $\mu_s$ exist
if and only if $p\,\big|\gcd(n,r)$ and
$\nu_p(s-1)<\min\{\nu_p(q-1),\nu_p(rn)\}$.
\end{Corollary}

For the remaining case of $m=p=2$ in Theorem~\ref{theorem2}, we
obtain the following consequence.

\begin{Corollary}\label{duadic}
Assume that $s\in{\Bbb Z}_{rn}^*\cap(1+r{\Bbb Z}_{rn})$. Then
Type I duadic splittings for $1+r{\Bbb Z}_{rn}$ given by $\mu_s$ exist
if and only if both $n$ and $r$ are even and one of the following four conditions holds:
\begin{itemize}
\item[\bf(i)]
$\nu_2(q-1)>|\nu_2(s-1)|$ and $\nu_2(rn)>|\nu_2(s-1)|$;
\item[\bf(ii)]
$\nu_2(q-1)=1$, $\nu_2(s-1)>1$,
 $\nu_{2}(q+1)+1>|\nu_2(s-1)|$ and $\nu_{2}(rn)>|\nu_2(s-1)|$;
\item[\bf(iii)]
$\nu_2(q-1)=\nu_2(s-1)=1$,~
 $|\nu_2(s+1)|>\nu_{2}(q+1)$ and $\nu_2(rn)>\nu_{2}(q+1)$;
\item[\bf(iv)]
$\nu_2(q-1)=\nu_2(s-1)=1$,~
 $|\nu_2(s+1)|<\nu_{2}(q+1)$ and $|\nu_2(s+1)|<\nu_{2}(rn)$.
\end{itemize}
\end{Corollary}

\begin{proof}
By Theorem~\ref{theorem2},
we need to look for a condition such that $\nu_2(M_s)\ge 1$.
If $\nu_2(q-1)\ge 2$, by (i) and (iii) of Theorem~\ref{theorem2},
we arrive at (i) of the corollary. Furthermore, (ii) of the corollary
follows from (ii) of Theorem~\ref{theorem2},
while (iii) and (iv) of the corollary
follow from (iv) of Theorem~\ref{theorem2}.
\end{proof}

\begin{Remark}\rm
Note that, if $s\in{\Bbb Z}_{rn}^*\cap(1+r{\Bbb Z}_{rn})$ and $r$ is even,
then $s$ is odd, i.e., $|\nu_2(s-1)|\ge 1$;
hence Condition (i) of Corollary~\ref{duadic} implies that
$\nu_2(q-1)\ge 2$, or equivalently, $q\equiv 1~({\rm mod}~4)$.
Hence, Corollary~\ref{duadic}(i) yields again, but in different notations,
the result \cite[Theorem 20]{Blackford13} for the case
when $q\equiv 1~({\rm mod}~4)$. Moreover,
a special case of  Corollary~\ref{duadic} (iii) and (iv)
was also described in \cite[Theorem 20]{Blackford13},
which, however, contains some inaccuracies.
A correction to \cite[Theorem 20]{Blackford13} has been shown
in \cite[Theorem 1.3]{Chen} as follows:

\medskip
\begin{quote}
{\it
Assume $q\equiv3~(\bmod~4)$,
with $q=-1+2^cd$ for some $c\geq2$ and some odd $d$.
Let $r=2r'$,  $n=2^bn'$ and $s=1+2r'n'$, with $r', n'$ odd and $b\geq2$.
\begin{itemize}

\item[{\rm(A)}]~ $\mu_s$ is a Type I duadic splitting
for $1+r{\Bbb Z}_{rn}$ if and only if
one of the following conditions holds:
{\rm (1)}~$c>b>\nu_2(1+r'n')$; {\rm (2)}~$b\geq c>\nu_2(1+r'n')$.

\item[{\rm(B)}]~For $2\leq i<1+b$, $\mu_{1+2^ir'n'}$  is a
Type I duadic splitting  for $1+r{\Bbb Z}_{rn}$ if and only if $i\leq c$.

\end{itemize}}
\end{quote}

\medskip\noindent One can see that  statement (B) follows from
(ii) of Corollary~\ref{duadic}, while statement (A) follows from
(iii) and (iv) of the corollary.
Moreover, (ii), (iii) and (iv) of Corollary~\ref{duadic}
are more extensive than the result of~\cite{Chen} stated above,
e.g.,  the case ``$s=-1$'' does not appear in
\cite[Theorem 1.3]{Chen} but is included in Corollary~\ref{duadic}:
since $\nu_2((-1)-1)=1$ and $|\nu_2((-1)+1)|=\infty$,
the following corollary follows at once.
\end{Remark}

\begin{Corollary}\label{s=-1}
Type I duadic splittings for $1+r{\Bbb Z}_{rn}$ given by $\mu_{-1}$ exist
if and only if $n$ is even, $r=2$ and one of the following two conditions holds:
\begin{itemize}

\item[\bf(i)]
 $\nu_2(q-1)\ge 2$ (i.e., $q\equiv 1~({\rm mod}~4)$);

\item[\bf(ii)]
 $\nu_2(q-1)=1$ (i.e., $q\equiv 3~({\rm mod}~4)$) and $\nu_2(q+1)<\nu_2(rn)$.

\end{itemize}
\end{Corollary}

As mentioned in \cite{Blackford08},
Euclidean self-dual negacyclic codes are just
Type I duadic negacyclic codes given by $\mu_{-1}$.
In this sense, Corollary \ref{s=-1}
is just \cite[Theorem 3]{Blackford08}.

\subsection{Alternant constacyclic MDS codes}

By an {\em alternant code}, we mean a subfield subcode of a GRS code
${\rm GRS}_k(\mbox{\boldmath$\alpha$};{\mathbf v})$
over a large field ${\Bbb F}_{q^e}$, i.e.,
the code over the ground field ${\Bbb F}_q$, denoted by
${\rm GRS}_k(\mbox{\boldmath$\alpha$};{\mathbf v})|_{{\Bbb F}_q}$,
obtained by restricting the GRS code
${\rm GRS}_k(\mbox{\boldmath$\alpha$};{\mathbf v})$
over ${\Bbb F}_{q^e}$ to ${\Bbb F}_q$
(cf. \cite[Ch.~9]{Lingbook}).

For the case (i) of Corollary~\ref{s=-1},
we have shown in Example~\ref{ex1} a family of self-dual
negacyclic GRS codes with parameters
$[\frac{q-1}{2},\frac{q-1}{4},\frac{q+3}{4}]$.
On the other hand, it is easy to see that
there are no self-dual negacyclic GRS codes for the case
(ii) of Corollary~\ref{s=-1}: since $2n\nmid(q-1)$, there are no
primitive $2n$th roots of unity in ${\Bbb F}_q$.
However, Proposition~\ref{p-adic} and Example~\ref{ex1}
provide a way to construct self-dual negacyclic alternant
MDS codes for both the cases of Corollary~\ref{s=-1}.

\begin{Proposition}\label{duadic alter}
Assume that $q$ is a power of an odd prime. Let $n=\frac{q+1}{\ell}$
with $\ell$ being an odd divisor of $q+1$,
let $\omega\in{\Bbb F}_{q^2}$ be a primitive $2n$th root of unity, and let
$$\textstyle
{\cal X}_0=\{1,\,3,\,\cdots,\,n-1\},\qquad
 {\cal X}_1=\{n+1,\,n+3,\,\cdots,\,2n-1\},
$$
as in Eqn~(\ref{dual GRS}). Then
\begin{itemize}

\item[\bf(i)]
${\cal X}_0$ and ${\cal X}_1$ form
a Type I duadic splitting of $1+2{\Bbb Z}_{2n}$ over ${\Bbb F}_{q}$
given by $\mu_{-1}$;

\item[\bf(ii)] the duadic negacyclic codes
$C_{{\cal X}_0}, C_{{\cal X}_1}$ over ${\Bbb F}_{q}$
are self-dual duadic negacyclic MDS $[n,\frac{n}{2},\frac{n}{2}+1]$ codes;

\item[\bf(iii)]
$C_{{\cal X}_0}=
{\rm GRS}_{n/2}(\mbox{\boldmath$\omega$};{\mathbf v})|_{{\Bbb F}_q}$
is an alternant code,
where  ${\rm GRS}_{n/2}(\mbox{\boldmath$\omega$};{\mathbf v})$
is the GRS code over ${\Bbb F}_{q^2}$ with
$\mbox{\boldmath$\omega$}=(1,\omega^{-2},\cdots,\omega^{-2(n-1)})$
and ${\mathbf v}=(1,\omega^{-1},\cdots,\omega^{-(n-1)})$.

\end{itemize}
\end{Proposition}

\begin{proof}
Note that, for any odd integer $t$, we have $tn\equiv n~({\rm mod}~2n)$.
Since $q=\ell n -1$ with $\ell$ being odd,
we have $q\equiv n-1~({\rm mod}~2n)$.
For any $i\in{\cal X}_0$, since $i$ is odd, we have
$$qi\equiv (n-1)i=ni-i\equiv n-i~({\rm mod}~2n).$$
Thus $\mu_q({\cal X}_0)={\cal X}_0$, i.e., both
${\cal X}_0$ and ${\cal X}_1$ are $\mu_q$-invariant,
which proves the conclusion (i).

By Proposition~\ref{p-adic} and Example~\ref{ex1},
the duadic negacyclic code $\tilde C_{{\cal X}_0}$
over ${\Bbb F}_{q^2}$ is a self-dual negacyclic GRS code
as follows:
\begin{eqnarray*}
\tilde C_{{\cal X}_0}&=&{\rm GRS}_{n/2}(\mbox{\boldmath$\omega$};{\mathbf v})\\
&=&\Big\{\Big(f(1),\,\omega^{-1}f(\omega^{-2}),
 \,\cdots,\,\omega^{-(n-1)}f(\omega^{-2(n-1)})\Big)\,\Big|\,
  f(X)\in{\Bbb F}_{q^2}[X],~ \deg f(X)<\frac{n}{2}\Big\}.
\end{eqnarray*}
Note that $C_{{\cal X}_0}\subseteq \tilde C_{{\cal X}_0}$
and $\omega C_{{\cal X}_0}\subseteq \tilde C_{{\cal X}_0}$, and that
$\dim_{\mathbb{F}_q}\tilde C_{{\cal X}_0}
 =2\dim_{\mathbb{F}_{q^2}}\tilde C_{{\cal X}_0}=n$.
We have the direct sum
$\tilde C_{{\cal X}_0} = C_{{\cal X}_0} \oplus \omega C_{{\cal X}_0}$.
Therefore, $C_{{\cal X}_0}=\tilde C_{{\cal X}_0}|_{{\Bbb F}_q}$
is the desired subfield subcode of the code~$\tilde C_{{\cal X}_0}$.
Both the conclusions (ii) and (iii) now follow easily.
\end{proof}

\begin{Remark}\rm The biggest choice of $n$ in
Proposition~\ref{duadic alter} is $n=q+1$.
For this choice, the self-dual duadic negacyclic
alternant MDS code $C_{{\cal X}_0}$ has parameters
$[q+1,\frac{q+1}{2},\frac{q+3}{2}]$.
Blackford \cite[Corollary 5]{Blackford08} has constructed
this self-dual negacyclic $[q+1,\frac{q+1}{2},\frac{q+3}{2}]$ code,
but did not show it to be an alternant code.
\end{Remark}

\begin{Proposition}\label{ex2}
Let $q$ be an odd prime power such that $\nu_2(q-1)\ge 3$.
Let $r=\frac{q-1}{2}$, $s=1+\frac{q^2-1}{4}$,  $n=q+1$,
let $\omega\in{\Bbb F}_{q^2}$ be a primitive $rn${\rm th} root of unity  and let
$$
{\cal X}_0=\Big\{1+\frac{q-1}{2}j\,\Big{|}\,
  -\frac{q-1}{4}<j\le\frac{q-1}{4}+1\Big\},~~~
  {\cal X}_1=(1+r\mathbb{Z}_{rn})\setminus{\cal X}_0.
$$
Then
\begin{itemize}

\item[(i)]
${\cal X}_0$ and ${\cal X}_1$ form
a Type I duadic splitting of $1+r{\Bbb Z}_{rn}$ over ${\Bbb F}_{q}$
given by $\mu_{s}$;

\item[(ii)] the duadic constacyclic codes $C_{{\cal X}_0}$ and
$C_{{\cal X}_1}$ over ${\Bbb F}_{q}$
are MDS $[n,\frac{n}{2},\frac{n}{2}+1]$ codes;

\item[(iii)] $C_{{\cal X}_1}=
{\rm GRS}_{n/2}(\mbox{\boldmath$\omega$};{\mathbf v})|_{{\Bbb F}_q}$
is an alternant code,
where  ${\rm GRS}_{n/2}(\mbox{\boldmath$\omega$};{\mathbf v})$
is the GRS code over ${\Bbb F}_{q^2}$ with
$\mbox{\boldmath$\omega$}=\big(1,\omega^{r},\cdots,\omega^{(n-1)r}\big)$
and ${\bf v}=
\big(1,\omega^{\frac{q-1}{4}r-1},\omega^{\frac{q-1}{4}2r-2},
  \cdots,\omega^{\frac{q-1}{4}(n-1)r-(n-1)}\big)$.

\end{itemize}
\end{Proposition}

\begin{proof}
It is clear that $s\in{\Bbb Z}_{rn}^*\cap(1+r{\Bbb Z}_{rn})$.
To prove (i), it is enough to
show  that ${\cal X}_0$ is a union of some $q$-cyclotomic cosets
modulo $rn$ with $|{\cal X}_0|=\frac{q+1}{2}$ and
$s{\cal X}_0\bigcap {\cal X}_0=\emptyset$.
Clearly, $|{\cal X}_0|=\frac{q+1}{2}$.
To show that ${\cal X}_0$ is a union of some $q$-cyclotomic cosets,
it suffices to prove that $q(1+\frac{q-1}{2}j)\in {\cal X}_0$ for any
$-\frac{q-1}{4}<j\leq\frac{q-1}{4}+1$.
This is straightforward:
$q(1+\frac{q-1}{2}j)\equiv~1+\frac{q-1}{2}(2-j)~(\bmod~rn)$
and $-\frac{q-1}{4}+1\leq2-j<\frac{q-1}{4}+2$.
We are left to show that $s{\cal X}_0\bigcap {\cal X}_0=\emptyset$.
Assuming otherwise, then two integers $j, j'$
with $-\frac{q-1}{4}<j,j'\leq\frac{q-1}{4}+1$ can be found such that
$1+\frac{q-1}{2}j\equiv 1+\frac{q-1}{2}(j'+\frac{q+1}{2})
~(\bmod~\frac{q^2-1}{2})$.
We then have $j-j'\equiv \frac{q+1}{2}~(\bmod~q+1)$, which is impossible.
Thus $\{{\cal X}_0, ~{\cal X}_1\}$ is a splitting of $1+r{\Bbb Z}_{rn}$
given by $\mu_s$, proving (i).

Observe that ${\rm ord}_{rn}(q)=2$.
Let $\tilde C_{{\cal X}_1}$ be the constacyclic code
of length $q+1$ over $\mathbb{F}_{q^2}$ with check polynomial
$\prod_{Q\in{\cal X}_1/\mu_q}M_Q(X)$.
Hence, $\{\omega^{j}\mid j\in{\cal X}_0\}$ is the  set of zeros
of the code $\tilde C_{{\cal X}_1}$.
Using reasoning similar to that in the proof of  Proposition \ref{p-adic},
one gets
\begin{equation*}
\tilde C_{{\cal X}_1}=\big\{\big(f(1),
\omega^{\frac{q-1}{4}r-1}f(\omega^{r}), \cdots,
\omega^{\frac{q-1}{4}(n-1)r-(n-1)}f(\omega^{(n-1)r})\big)\,|\,
f(X)\in \mathbb{F}_{q^2}[X]~~\hbox{and $\deg f(X)<\frac{n}{2}$}\big\}.
\end{equation*}
It is easy to see that
$\tilde C_{{\cal X}_1}\bigcap \mathbb{F}_q^n=C_{{\cal X}_1}$.
We are done.
\end{proof}

We conclude this discussion with some examples in Table 5.2.
The alternant codes (i)-(iii) correspond to the codes
(iv)-(vi) of Table~\ref{t1}, respectively,
using Proposition~\ref{duadic alter}, while the alternant codes (iv)-(v)
are derived from Proposition \ref{ex2}.

\smallskip
\begin{table}[h]\label{t2}
\begin{center}
{\footnotesize\renewcommand{\arraystretch}{2}
\begin{tabular}{|c|c|c|c|c|c|c|c|}\hline
No & $m$ & $q$ & $n$ & $r$ & alternant code & parameters\\ \hline
(i) & 2 & $3^2$ & 10  & 2 &
 $\big\{(f(1),\omega^{-1}f(\omega^{-2}),\cdots,\omega^{-9}f(\omega^{-18}))
 \,\big|\,f(X)\in \mathbb{F}_{3^4}[X],~\deg f(X)<5\big\}
 \big|_{\mathbb{F}_{3^2}}$   &[10,5,6] \\
(ii) & 2 & $5$ & 6  & 2 &
 $\big\{(f(1),\omega^{-1}f(\omega^{-2}),\cdots,\omega^{-5}f(\omega^{-10}))
 \,\big|\,f(X)\in \mathbb{F}_{5^2}[X],~\deg f(X)<3\big\}
 \big|_{\mathbb{F}_{5}}$   &[6,3,4] \\
(iii) & 2 & $7$ & 8  & 2 &
 $\big\{(f(1),\omega^{-1}f(\omega^{-2}),\cdots,\omega^{-7}f(\omega^{-14}))
 \,\big|\,f(X)\in \mathbb{F}_{7^2}[X],~\deg f(X)<4\big\}
 \big|_{\mathbb{F}_{7}}$   &[8,4,5] \\
(iv) & 2 & $3^2$ & 10  & 4 &
 $\big\{(f(1),\omega^{7}f(\omega^{4}),\cdots,\omega^{63}f(\omega^{36}))
 \,\big|\,f(X)\in \mathbb{F}_{9^2}[X],~\deg f(X)<5\big\}
 \big|_{\mathbb{F}_{9}}$   &[10,5,6] \\
 (v) & 2 & $17$ & 18  & 8 &
 $\big\{(f(1),\omega^{33}f(\omega^{8}),\cdots,\omega^{527}f(\omega^{136}))
 \,\big|\,f(X)\in \mathbb{F}_{17^2}[X],~\deg f(X)<9\big\}
 \big|_{\mathbb{F}_{17}}$   &[18,9,10] \\
\hline
\end{tabular}}
\begin{caption}
{Alternant codes from Type {\rm I} duadic constacyclic codes}
\end{caption}
\end{center}
\end{table}

\section*{Acknowledgements}

The main results of this work were obtained while the third author
was visiting the fourth author at Nanyang Technological University
in Jan-Feb 2014.
He is grateful for the hospitality and the support.
The research of Bocong Chen, Yun Fan and San Ling is supported
by NSFC with grant numbers 11271005 and 11171370.
 The research of
Bocong Chen and San Ling is also
partially supported by Nanyang Technological University's research
grant number M4080456.


\begin{thebibliography}{99}

\bibitem{AB} J. L. Alperin, R. B. Bell, Groups and Representations,
GTM 162, Springer-Verlag, New York,  1997.

\bibitem{Aly}
S.A. Aly, A. Klappenecker, P.K.  Sarvepalli, Duadic group algebra codes,
In: Proc. Int. Symp. Inf. Theory, Adelaide, Australia, (2007),  2096-2100.

\bibitem{Aydin}
N. Aydin, I. Siap, D.J. Ray-Chaudhuri, The structure of $1$-generator quasi-twisted codes and new linear codes,
Des. Codes Cryptogr., {\bf 24}(2001),  313-326.


\bibitem{Berlekamp}
E.R. Berlekamp, Goppa codes, IEEE Trans. Inform. Theory,  {\bf 5}(1973),
590-592.

\bibitem{Blackford08}
T. Blackford, Negacyclic duadic codes,  Finite Fields Appl.,  {\bf 14}(2008),  930-943.


\bibitem{Blackford13}
T. Blackford, Isodual constacyclic codes,
Finite Fields Appl., {\bf 24}(2013),
29-44.


\bibitem{Brualdi}
R.A. Brualdi,  V. Pless, Polyadic codes,  Discr. Appl. Math., {\bf 25}(1989), 3-17.


\bibitem{Chen}
B. Chen, H.Q. Dinh, A note on isodual constacyclic codes,
 Finite Fields Appl., {\bf 29}(2014), 243-246.

\bibitem{CFLL} B. Chen, Y. Fan, L. Lin, H. Liu,
Constacyclic codes over finite fields,
Finite Fields Appl., {\bf 18}(2012), 1217-1231.

\bibitem{Ding2}
C. Ding, K.Y. Lam, C. Xing, Enumeration and construction of all duadic codes of length $p^m$, Fund. Inform.,
{\bf 38}(1999),  149-161.

\bibitem{Ding}
C. Ding, V. Pless, Cyclotomy and duadic codes of prime lengths, IEEE Trans. Inform. Theory, {\bf 45}(1999),
453-466.



\bibitem{Han}
S. Han, J.-L. Kim, Computational results of duadic double circulant codes, J.  Appl.  Math. Comput.,  {\bf 40}(2012),  33-43.


\bibitem{Huffman}
W.C. Huffman, V. Pless, Fundamentals of Error-Correcting Codes,
Cambridge University Press, Cambridge, 2003.

\bibitem{Leon}
J.S. Leon, J.M. Masley,  V. Pless, Duadic codes,
IEEE Trans. Inform. Theory,  {\bf 30}(1984),
709-714.


\bibitem{Lim}
C. J. Lim, Consta-abelian polyadic codes, IEEE Trans. Inform. Theory,  {\bf 51}(2005),
2198-2206.

\bibitem{Lingbook}
S. Ling,  C. Xing,  Coding Theory: A First Course, Cambridge University Press, Cambridge,   2004.


\bibitem{Ling04}
S. Ling, C. Xing, Polyadic codes revisited, IEEE Trans. Inform. Theory,  {\bf 50}(2004),
200-207.

\bibitem{Liu}
C. Liu, Some special classes of constacyclic codes (in Chinese),
Thesis (M. S.), Central China Normal University, 2010.


\bibitem{Pless}
V. Pless, Duadic codes revisited,  Congressus Numeratium, {\bf 59}(1987),
225-233.

\bibitem{Pless2}
V. Pless,  J.J. Rushanan, Triadic codes,  Linear Algebra Appl., {\bf 98}(1988),
415-433.


\bibitem{Rushanan}
J.J. Rushanan,   Duadic codes and difference sets,  J. Combin. Theory Ser. A,  {\bf 57}(1991), 254-61.

\bibitem{Sharma}
A.  Sharma, G.K. Bakshi,  M.  Raka, Polyadic codes of prime power length,
Finite Fields Appl., {\bf 13}(2007),
1071-1085.


\bibitem{Smid}
M.H.M. Smid, Duadic codes, IEEE Trans. Inform. Theory,  {\bf 33}(1987),  432-433.


\bibitem{van}
J.H. van Lint,  Introduction to Coding Theory, Springer, Berlin,  1982.

\bibitem{Ward}
H.N. Ward, L. Zhu,  Existence of abelian group codes partitions,  J. Combin. Theory Ser. A,  {\bf67}(1994),
276-281.

\bibitem{Yang}
Y. Yang,  W. Cai, On self-dual constacyclic codes over finite fields, Des. Codes Cryptogr.,
DOI 10.1007/s10623-013-9865-9.

\bibitem{Zhang}
S. Zhang, Existence of certain class of duadic group algebra codes, J. Statist. Plann. Inference,  {\bf94}(2001), 405-411.

\bibitem{Zhu}
L. Zhu, Duadic group algebra codes, J. Statist. Plann. Inference, {\bf51}(1996), 395-401.

\end{thebibliography}
\end{document}